\documentclass[5p, review]{elsarticle}

\usepackage{amsthm}
\usepackage{amsmath}
\usepackage[ruled, vlined, linesnumbered, nofillcomment]{algorithm2e}
\usepackage{tikz}
\usepackage{pgfplots}
\usepackage{natbib}

\newcommand{\set}[1]{\left\{#1\right\}}
\newcommand{\pr}[1]{\left(#1\right)}
\newcommand{\fpr}[1]{\mathopen{}\left(#1\right)}

\newcommand{\abs}[1]{{\left|#1\right|}}

\newcommand{\define}{\leftarrow}

\newcommand{\bigO}[1]{\dispfunc{\mathcal{O}}{#1}}

\DeclareRobustCommand{\dispfunc}[2]{%
  \ensuremath{%
  \ifthenelse{\equal{#2}{}}%
    {\mathit{#1}}%
    {\mathit{#1}\fpr{#2}}}}

\newcommand{\score}[1]{\dispfunc{p}{#1}}
\newcommand{\scoremax}[1]{\dispfunc{p_\textit{max}}{#1}}

\newcommand{\prbseg}{\textsc{Seg}\xspace}
\newcommand{\prbmaxseg}{\textsc{MaxSeg}\xspace}
\newcommand{\prballmaxseg}{\textsc{AllMaxSeg}\xspace}
\newcommand{\prballseg}{\textsc{AllSeg}\xspace}

\newcommand{\algfmslow}{\textsc{all-ms}\xspace}

\newcommand{\algoracle}{\textsc{oracle}\xspace}
\newcommand{\algestimate}{\textsc{estimate}\xspace}
\newcommand{\algsparsify}{\textsc{sparsify}\xspace}
\newcommand{\algalldp}{\textsc{all-dp}\xspace}

\newtheorem{lemma}{Lemma}[section]
\newtheorem{proposition}{Proposition}[section]

\newtheorem{problem}{Problem}[section]

\SetKwComment{tcpas}{\{}{\}}
\SetCommentSty{textnormal}
\SetArgSty{textnormal}
\SetKw{False}{false}
\SetKw{True}{true}
\SetKw{Null}{null}
\SetKwInOut{Output}{output}
\SetKwInOut{Input}{input}
\SetKw{AND}{and}
\SetKw{OR}{or}
\SetKw{Break}{break}

\let\oldnl\nl
\newcommand{\nonl}{\renewcommand{\nl}{\let\nl\oldnl}}

\pgfdeclarelayer{background}
\pgfdeclarelayer{foreground}
\pgfsetlayers{background,main,foreground}

\definecolor{yafaxiscolor}{rgb}{0.3, 0.3, 0.3}

\definecolor{yafcolor1}{rgb}{0.4, 0.165, 0.553}
\definecolor{yafcolor2}{rgb}{0.949, 0.482, 0.216}
\definecolor{yafcolor3}{rgb}{0.47, 0.549, 0.306}
\definecolor{yafcolor4}{rgb}{0.925, 0.165, 0.224}
\definecolor{yafcolor5}{rgb}{0.141, 0.345, 0.643}
\definecolor{yafcolor6}{rgb}{0.965, 0.933, 0.267}
\definecolor{yafcolor7}{rgb}{0.627, 0.118, 0.165}
\definecolor{yafcolor8}{rgb}{0.878, 0.475, 0.686}

\tikzstyle{exnode} = [inner sep = 1pt]
\tikzstyle{labnode} = [sloped, text = black, font = \scriptsize, inner sep = 1pt]
\tikzstyle{exedge} = [yafcolor5, draw, thick, >=latex, ->]
\tikzstyle{exedge2} = [yafcolor2, draw, thick, >=latex, ->]

\tikzstyle{interval} = [yafcolor1, very thick]

\newlength{\intervalwidth}
\newlength{\yafaxispad}
\newlength{\yaftlpad}
\newlength{\yaflabelpad}
\newlength{\yafaxiswidth}
\newlength{\yafticklen}
\makeatletter
\def\pgfplots@drawtickgridlines@INSTALLCLIP@onorientedsurf#1{}
\makeatother

\newcommand{\yafdrawxaxis}[2]{
	\pgfplotstransformcoordinatex{#1}\let\xmincoord=\pgfmathresult 
	\pgfplotstransformcoordinatex{#2}\let\xmaxcoord=\pgfmathresult 
	\pgfsetlinewidth{\yafaxiswidth} 
	\pgfsetcolor{yafaxiscolor}
	\pgfpathmoveto{\pgfpointadd{\pgfpointadd{\pgfplotspointrelaxisxy{0}{0}}{\pgfqpointxy{\xmincoord}{0}}}{\pgfqpoint{-0.5\yafaxiswidth}{\yafaxispad}}}
	\pgfpathlineto{\pgfpointadd{\pgfpointadd{\pgfplotspointrelaxisxy{0}{0}}{\pgfqpointxy{\xmaxcoord}{0}}}{\pgfqpoint{0.5\yafaxiswidth}{\yafaxispad}}}
	\pgfusepath{stroke}

}
\newcommand{\yafdrawyaxis}[2]{
	\pgfplotstransformcoordinatey{#1}\let\ymincoord=\pgfmathresult 
	\pgfplotstransformcoordinatey{#2}\let\ymaxcoord=\pgfmathresult 
	\pgfsetlinewidth{\yafaxiswidth} 
	\pgfsetcolor{yafaxiscolor}
	\pgfpathmoveto{\pgfpointadd{\pgfpointadd{\pgfplotspointrelaxisxy{0}{0}}{\pgfqpointxy{0}{\ymincoord}}}{\pgfqpoint{\yafaxispad}{-0.5\yafaxiswidth}}}
	\pgfpathlineto{\pgfpointadd{\pgfpointadd{\pgfplotspointrelaxisxy{0}{0}}{\pgfqpointxy{0}{\ymaxcoord}}}{\pgfqpoint{\yafaxispad}{0.5\yafaxiswidth}}}
	\pgfusepath{stroke}
}

\pgfplotscreateplotcyclelist{yaf}{%
{yafcolor1,mark options={scale=0.75},mark=o}, 
{yafcolor2,mark options={scale=0.75},mark=square},
{yafcolor3,mark options={scale=0.75},mark=triangle},
{yafcolor4,mark options={scale=0.75},mark=o},
{yafcolor5,mark options={scale=0.75},mark=o},
{yafcolor6,mark options={scale=0.75},mark=o},
{yafcolor7,mark options={scale=0.75},mark=o},
{yafcolor8,mark options={scale=0.75},mark=o}} 

\pgfplotsset{axis y line=left, axis x line=bottom,
	tick align=outside,
	compat = 1.3,
	tickwidth=\yafticklen,
	clip = false,
	every axis title shift = 0pt,
    x axis line style= {-, line width = 0pt, opacity = 0},
    y axis line style= {-, line width = 0pt, opacity = 0},
    x tick style= {line width = \yafaxiswidth, color=yafaxiscolor, yshift = \yafaxispad},
    y tick style= {line width = \yafaxiswidth, color=yafaxiscolor, xshift = \yafaxispad},
    x tick label style = {font=\scriptsize, yshift = \yaftlpad},
    y tick label style = {font=\scriptsize, xshift = \yaftlpad},
    every axis y label/.style = {at = {(ticklabel cs:0.5)}, rotate=90, anchor=center, font=\scriptsize, yshift = -\yaflabelpad},
    every axis x label/.style = {at = {(ticklabel cs:0.5)}, anchor=center, font=\scriptsize, yshift = \yaflabelpad},
    x tick label style = {font=\scriptsize, yshift = 1pt},
    grid = major,
    major grid style  = {dash pattern = on 1pt off 3 pt},
	every axis plot post/.append style= {line width=\yafaxiswidth} ,
	legend cell align = left,
	legend style = {inner sep = 1pt, cells = {font=\scriptsize}},
	legend image code/.code={%
		\draw[mark repeat=2,mark phase=2,#1] 
		plot coordinates { (0cm,0cm) (0.15cm,0cm) (0.3cm,0cm) };%
	} 
}

\begin{document}

\begin{frontmatter}
\title{Strongly polynomial efficient approximation scheme for segmentation}
\author{Nikolaj Tatti}
\address{F-Secure, HIIT, Aalto University, Finland}
\ead{nikolaj.tatti@aalto.fi}

\begin{abstract}
Partitioning a sequence of length $n$ into $k$ coherent segments (\prbseg) is one of the classic
optimization problems. As long as the optimization criterion is additive, \prbseg
can be solved exactly in $\bigO{n^2k}$ time using a classic dynamic
program.  Due to the quadratic term, computing the exact segmentation may be
too expensive for long sequences, which has led to development of approximate
solutions. We consider an existing estimation scheme that computes $(1 + \epsilon)$
approximation in polylogarithmic time. We augment this algorithm, making it
strongly polynomial. We do this by first solving a slightly different
segmentation problem (\prbmaxseg), where the quality of the segmentation is the maximum
penalty of an individual segment. By using this solution to initialize the estimation
scheme, we are able to obtain a strongly polynomial algorithm.
In addition, we consider a cumulative version
of \prbseg, where we are asked to discover the optimal segmentation for each prefix of the input sequence.
We propose a strongly polynomial algorithm that yields $(1 + \epsilon)$
approximation in $\bigO{nk^2 / \epsilon}$ time. Finally, we consider a cumulative
version of \prbmaxseg, and show that we can solve the problem in $\bigO{nk \log k}$ time.
\end{abstract}
\end{frontmatter}



\section{Introduction}
Partitioning a sequence into coherent segments is one of the classic
optimization problems, with applications in various domains, such as
discovering context in mobile devices~\citep{himberg:01:time}, similarity
search in time-series databases~\citep{keogh2001locally}, and bioinformatics~\citep{li2001dna,salmenkivi2002genome}.

More formally, we are given a sequence of length $n$ and
a penalty function of a segment, and we are asked to find a segmentation with $k$
segments such that the sum of penalties is minimized (\prbseg). As long as the score is
additive, \prbseg can be solved exactly in $\bigO{n^2k}$ time using a
classic dynamic program~\citep{bellman:61:on}.

Due to the quadratic term, computing the exact segmentation may be too expensive for long sequences,
which has led to development of approximate
solutions. \citet{guha:06:estimate} suggested an algorithm
that yields $(1 + \epsilon)$ approximation in $\bigO{k^3 \log^2 n + k^3
\epsilon^{-2} \log n}$ time, if we can compute the penalty of a single segment
in constant time.\footnote{Note that \citet{guha:06:estimate} refers to this problem as histogram construction.} This method assumes that we are dealing with an integer
sequence that is penalized by $L_2$-error. Without these assumptions 
the computational complexity deteriorates to
$\bigO{k^3 \log \frac{\theta}{\alpha} \log n + k^3 \epsilon^{-2} \log n}$, 
where $\theta$ is the cost of the optimal solution and $\alpha$ is the smallest possible
non-zero penalty.
Consequently, without the aforementioned assumptions the term $\bigO{\frac{\theta}{\alpha}}$ may be arbitrarily
large, and the method is not strongly polynomial.

In this paper we demonstrate a simple approach for how to augment this method
making it strongly polynomial. The reason for having $\log
\frac{\theta}{\alpha}$ term is that the algorithm needs to first find a
2-approximation. This is done by setting $\alpha$ to be the smallest non-zero
cost and then increasing exponentially its value until an appropriate value is
discovered (we can verify whether the value is appropriate in $\bigO{k^3 \log n }$ time).
Instead of using the smallest non-zero cost, we first compute a $k$-approximation
of the segmentation cost, say $\eta$, and then set $\alpha = \eta / k$.
This will free us of integrality assumptions, reducing the computation time to
$\bigO{k^3 \log k \log n + k^3 \epsilon^{-2} \log n}$. Moreover, we no longer need to
discover the smallest non-zero cost, which can be non-trivial.

In order to discover $k$-approximation we consider a different segmentation
problem, where the score of the whole segmentation is not the sum but the
maximum value of a single segment (\prbmaxseg). We show in Section~\ref{sec:seg} that the segmentation
solving \prbmaxseg yields the needed $k$-approximation for \prbseg. Luckily, we can
solve \prbmaxseg in $\bigO{k^2 \log^2 n}$ time by
using an algorithm by~\citet{guha:07:linear}.

The method by~\citet{guha:06:estimate} only computes a $k$-segmentation for whole
sequence. We consider a cumulative variant of the segmentation problem, where we are asked to compute an
$\ell$-segmentation for all prefixes and for all $\ell \leq k$.
As our second contribution, given in Section~\ref{sec:allseg}, we propose a \emph{strongly polynomial} algorithm
that yields $(1 + \epsilon)$ approximation in $\bigO{nk^2 / \epsilon}$ time.
Finally, in Section~\ref{sec:allmaxseg} we also consider a cumulative variant of \prbmaxseg problem,
for which we propose an exact algorithm with computational complexity of $\bigO{nk \log k}$.

\section{Related work}\label{sec:related}

As discussed earlier, our approach is based on improving method given by
\citet{guha:06:estimate}, that achieves an $(1 + \epsilon)$ approximation in
$\bigO{k^3 \log^2 n + k^3 \epsilon^{-2} \log n}$ time, under some assumptions.
\citet{terzi:06:efficient} suggested an approximation algorithm
that yields $3$-approximation in $\bigO{n^{4/3} k^{5/3}}$ time, assuming $L_2$
error.

We also consider an algorithm for the cumulative version of the problem.  The
main idea behind the algorithm is inspired 
by~\citet{guha:01:estimate}, where the algorithm requires
$\bigO{nk^2 / \epsilon \log n}$ time as well as integrality assumptions.  We
achieve $\bigO{nk^2 / \epsilon}$ time and, more importantly, we do not need any
integrality assumptions, making the algorithm strongly polynomial.

Fast heuristics that do not yield approximation guarantees have been proposed.
These methods include top-down approaches, based on splitting segments
(see~\citet{shatkay:96:approximate,bernaola-galvan:96:compositional,douglas:73:algorithms,lavrenko:00:mining},
for example), and bottom-up approaches, based
on merging segments 
(see~\citet{palpanas:04:online}, for example).
A different approach was suggested by~\citet{himberg:01:time}, where the authors optimize boundaries of a random
segmentation.

If the penalty function is concave,
then we can discover the exact optimal segmentation in $\bigO{k(n + k)}$ time
using the SMAWK algorithm~\citep{aggarwal1987geometric,galil:90:linear}. 
This is the case when segmenting a monotonic one-dimensional sequence
and using $L_1$-error as a penalty~\citep{hassin:91:improved,fleisher:06:online}.

If we were to evaluate the segmentation by the maximum penalty of a segment, instead of the sum of all penalties,
then the problem changes radically. \citet{guha:07:linear} showed that
we can compute the exact solution in $\bigO{n + k^2 \log^3 n}$ time,
when using the $L_\infty$ penalty.
\citet{DBLP:conf/vldb/GuhaSW04} also showed that we can compute the cumulative version in $\bigO{kn\log^2 n}$ time,\footnote{or in $\bigO{kn \log n}$ time if we assume
that we can compute the penalty of a segment in constant time.}
which we improve in Section~\ref{sec:allmaxseg}.

\section{Preliminaries}\label{sec:prel}

\paragraph{Segmentation problem}
Throughout the paper we will assume that we are given an integer $n$ and a \emph{penalty function}
$\score{}$ that maps two integers $1 \leq a \leq b \leq n$ to a positive real number.

The most common selection for the penalty function is an $L_q$
distance of individual points in a sequence segment to the optimal centroid, that is,
given a sequence of real numbers, $z_1, \ldots, z_{n + 1}$, the score is equal to
\begin{equation}
\label{eq:scoreex}
	L_q(a, b) = \min_{\mu} \sum_{i = a}^{b - 1} \| z_i  - \mu \|_q\quad .
\end{equation}
Since we do not need any notion of sequence in this paper, we abstract
it out, and speak directly only about penalty function.

Throughout the paper, we will assume that 
\begin{enumerate}
\item for any $1 \leq a_1 \leq a_2 \leq b_2 \leq b_1$, we have $\score{a_2, b_2} \leq \score{a_1, b_1}$. Also, $\score{a_1, a_1} = 0$,
\item we can compute the penalty in constant time,
\item we can perform arithmetic operations to the penalty in constant time, as well as compare the scores.
\end{enumerate}

The first assumption typically holds. For example, any $L_p$-error will satisfy
this assumption, as well as any log-likelihood-based errors~\citep{tatti:12:fast}.
The third assumption is a technicality used by the definition of strongly polynomial time.
The second assumption is the most limiting one. It holds for $L_2$-error: 
we can compute the penalty in constant time by precomputing cumulative mean and the second moment.
It also holds for log-linear models~\citep{tatti:12:fast}.
However, for example, we need $\bigO{\log n}$ time to compute $L_\infty$-error~\citep{guha:07:linear}.
In such a case, we need to multiply the running time by the time needed to compute the penalty.
We will ignore the running time needed for any precomputation as this depends on the used penalty.

Given an integer $k$ and an interval $[i, j]$, where $i$ and $j$ are integers with
$1 \leq i \leq j \leq n$, 
a $k$-\emph{segmentation} $B$ \emph{covering} $[i, j]$ is a sequence of $k + 1$
integers 
\[
	B = \pr{i = b_0 \leq b_1 \leq \cdots \leq b_k = j}\quad.
\]
We omit $k$ and simply write segmentation, whenever $k$ is known from context.

Let $\score{}$ be a penalty function for individual segments. Given a
segmentation $B$, we extend the definition of $\score{}$
and define $\score{B}$, the penalty for the segmentation $B$,
as
\[
	\score{B} = \sum_{i = 1}^k \score{b_{i - 1}, b_i}\quad.
\]

We can now state the classic segmentation problem: 

\begin{problem}[\prbseg]
Given a penalty $\score{}$ and an integer $k$ find a segmentation $B$ covering $[1, n]$
that minimizes $\score{B}$.
\end{problem}

Since the penalty score is the sum of the interval penalties, we can solve \prbseg with a dynamic program
given by~\citet{bellman:61:on}. The computational complexity of this program is $\bigO{n^2k}$, which
is prohibitively slow for large $n$.

We will also consider a cumulative variant of the segmentation problem defined as follows.

\begin{problem}[\prballseg]
Given a penalty function $\score{}$ and an integer $k$ find an \mbox{$\ell$-segmentation} $B$ covering $[1, i]$
that minimizes $\score{B}$, for every $i = 1, \ldots, n$ and $\ell = 1, \ldots k$.
\end{problem}

Note that~\citet{bellman:61:on} in fact solves \prballseg, and uses the solution to solve
\prbseg. However, the state-of-the-art approximation algorithm, by~\citet{guha:06:estimate}, solving \prbseg does not
solve \prballseg, hence we will propose a separate algorithm.

\paragraph{Approximation algorithm}
Our contribution is an additional component to the approximation algorithm
by~\citet{guha:06:estimate}, a state-of-the-art approximation
algorithm for estimating \prbseg. We devote the rest of this section
to explaining the technical details of this algorithm, and what is the issue that
we are addressing.

\citet{guha:06:estimate} showed that under some assumptions it is possible to
obtain a $(1 + \epsilon)$ approximation to \prbseg in $\bigO{k^3 \log^2(n) + k^3 \epsilon^{-2} \log(n)}$ time.
What makes this algorithm truly remarkable is that it is polylogarithmic in $n$, while the exact algorithm
is quadratic in $n$.
Here we assume that we have a constant-time access to the score $\score{}$.
If one uses $L_2$ error as a penalty, then one is forced to precompute the score, which requires
additional $\bigO{n}$ time. However, this term depends neither on $k$ nor on $\epsilon$.

The key idea behind the approach by~\citet{guha:06:estimate} is a sub-routine, $\algoracle(\delta, u)$,
that relies on two parameters $\delta$ and $u$. \algoracle constructs a
$k$-segmentation with the following property.\footnote{The actual subroutine is too complex to be described in compact space. For details, we refer reader to~\citep{guha:06:estimate}.}

\begin{proposition}
\label{prop:oracle}
Suppose that $\score{}$ is a penalty function, and $k$ is an integer.
Let $\theta$ be the cost of an optimal $k$-segmentation covering $[1, n]$.
Let $\delta$ and $u$ be two positive real numbers such that
$\theta + \delta \leq u$. Then $\algoracle(\delta, u)$ returns a $k$-segmentation
with a cost of $\tau$ such that $\tau \leq \theta + \delta$.
$\algoracle(\delta, u)$ runs in $\bigO{k^3 \frac{u^2}{\delta^2} \log n}$ time.
\end{proposition}

Proposition~\ref{prop:oracle} describes the trade-off between the accuracy and computational complexity:
We can achieve good accuracy with small $\delta$, and large enough $u$, but we have to pay the price in running time.

We will now describe how to select $u$ and $\delta$ in a smart way.
Let $\theta$ be the cost of an optimal $k$-segmentation.
Assume that we know an estimate of $\theta$, say $\eta$, such that $\eta \leq \theta \leq 2\eta$.
Let us set
\[
	u = (2 + \epsilon)\eta \quad\text{and}\quad \delta = \epsilon \eta\quad.
\]
As $\theta + \delta \leq u$,
Proposition~\ref{prop:oracle} guarantees that $\algoracle(\delta, u)$ returns a $k$-segmentation
with a cost of $\tau$ such that
\[
	\tau \leq \theta + \delta = \theta + \epsilon \eta \leq \theta + \epsilon \theta = (1 + \epsilon)\theta\quad.
\]
In other words, the resulting segmentation yields a $(1 + \epsilon)$ approximation guarantee.
The computational complexity of $\algoracle(\delta, u)$ is equal to $\bigO{k^3 \epsilon^{-2} \log n}$.

The difficult part is to discover $\eta$ such that $\eta \leq \theta \leq
2\eta$. This can be also done using the oracle (see Algorithm~\ref{alg:estimate}\footnote{
The original pseudo-code given by~\citet{guha:06:estimate} contains a small error, and only yields $\eta \leq \theta < 4\eta$.
Here we present the corrected variant.}). Assume that we know a \emph{lower bound} for $\theta$, say
$\alpha$. We first set $\eta = \alpha$ and check using \algoracle to see if $\eta$ is too small. If it is,
then we increase the value and repeat.

\begin{algorithm}
\caption{$\algestimate(\alpha)$, computes $\eta$ such that $\eta \leq \theta \leq 2 \eta$, where $\theta$ is the optimal cost.
Requires $\alpha \leq \theta$ as an input parameter.}
\label{alg:estimate}
$\eta \define \alpha$\;
$\tau \define $ the cost of the solution by $\algoracle(\eta / 2, 2\eta)$\;
\While {$\tau > 2 \eta$} {
	$\eta \define 1.5\eta$\;
	$\tau \define $ the cost of the solution by $\algoracle(\eta / 2, 2\eta)$\;
}
\Return $\eta$\;
\end{algorithm}

To see why \algestimate works,
assume that we are at a point in the while-loop where $2\eta < \theta$.
Then $\tau \geq \theta > 2\eta$ and the while-loop is not terminated.
This guarantees that for the final $\eta$, we have $\theta \leq 2\eta$.
Let us show that $\eta \leq \theta$. If the while-loop is terminated while $\eta < \theta$,
then there is nothing to prove. Assume otherwise, that is, at some point we have
$\eta \leq \theta \leq 1.5\eta$.
Since $\theta + \eta/2 \leq 2\eta$, Proposition~\ref{prop:oracle} guarantees that 
$\tau \leq \theta + \eta/2 \leq 2\eta$, and we exit the while-loop with $\eta \leq \theta$.

The computational complexity of a single \algoracle call is $\bigO{k^3 \log n}$,
and the total computational complexity of \algestimate is $\bigO{\log(\frac{\theta}{\alpha}) k^3 \log n}$.
At this point, \citet{guha:06:estimate} assume
(implicitly) that the penalty function is $L_2$ error of a sequence (see Eq.~\ref{eq:scoreex}),
and the values in the sequence are
integers encoded with a standard bit representation in at most $\bigO{\log n}$ space.
These assumptions have two consequences:
(\emph{i}) we can use $\alpha = 1/2$, the smallest non-zero cost for a segmentation, and 
(\emph{ii}) the number of iterations $\bigO{\log \frac{\theta}{\alpha}}$ is bounded by $\bigO{\log n}$.
This leads to a computational complexity of $\bigO{k^3 \log^2 n}$.

The $L_2$ assumption is not critical since the same argument can be done for
many other cost functions, however one is forced to find an appropriate
$\alpha$ for each case individually. Moreover, there are penalty functions for which
this argument does not work, for example, $\score{a, b} = \exp\pr{L_2(a, b)}$, where $L_2$ is given in Eq.~\ref{eq:scoreex}.

The more critical assumption is that the numbers in a sequence are integers,
and this assumption can be easily violated if we have a sequence of numbers
represented in a floating-point format. If this is the case, then we can no longer
select $\alpha = 1/2$. In fact, there is no easy way of selecting $\alpha$ such that
($i$)
we are sure that $\alpha$ is less than the optimal segmentation score, and
($ii$)
the number of loops needed by $\algestimate$,
$\bigO{\log \frac{\theta}{\alpha}}$, is bounded by a (slowly increasing) function of $n$.

In the following section, we will show how to select $\alpha$ such that the number
of loops in \algestimate remains small. More specifically, we demonstrate how to select
$\alpha$ such that $\alpha \leq \theta \leq k\alpha$. This immediately implies that the
computational complexity of \algestimate reduces to $\bigO{k^3 \log(k) \log(n)}$.
More importantly, we do not need any awkward assumptions about having a sequence of only integer values,
making this algorithm strongly polynomial.
Moreover, this procedure works on any penalty function, hence a finding 
appropriate $\alpha$ manually is no longer needed.

\section{Strongly polynomial scheme for segmentation}\label{sec:seg}

To find $\alpha$, the parameter for \algestimate,
we consider a different optimization problem, where
the segmentation is evaluated by its most costly segment.

\begin{problem}[\prbmaxseg]
\label{prb:max}
Given a penalty $\score{}$ and an integer $k$, 
find a $k$-segmentation $B$ covering $[1, n]$ that minimizes 
\[
	\scoremax{B} = \max_{1 \leq j \leq k} \score{b_{j - 1}, b_j}\quad.
\]
\end{problem}

The next proposition states why solving \prbmaxseg helps us to discover $\alpha$.

\begin{proposition}
\label{prop:estimate}
Suppose that $\score{}$ is a penalty function, and $k$ is an integer.
Let $B$ be a solution to \prbseg, and let $B'$ be a solution to \prbmaxseg.
Then
\[
	\score{B'}/k \leq \score{B} \leq \score{B'}\quad.
\]
\end{proposition}

\begin{proof}
To prove the first inequality write
\[
\begin{split}
	\score{B'} & = \sum_{i = 1}^k \score{b'_{i - 1}, b'_i} \leq k \max_{1 \leq i \leq k} \score{b'_{i - 1}, b'_i} \\
	& = k\scoremax{B'} \leq k\scoremax{B} \leq k\score{B}.
\end{split}
\]
The claim follows as the second inequality is trivial.
\end{proof}

By solving \prbmaxseg and obtaining a solution $B'$, we can set $\alpha =
\score{B} / k$ for \algestimate. The remaining problem is how to solve
\prbmaxseg in sub-linear time.

Here we can reuse an algorithm given by~\citet{guha:07:linear}.
This algorithm is designed to solve an instance of \prbmaxseg, where the penalty function is $L_\infty$, which
is
\[
	\score{a, b} = \min_{\mu} \sum_{i = a}^{b - 1} |x_i - \mu|\quad.
\]
Luckily, the same algorithm and the proof of correctness, see Lemma~3~in~\citep{guha:07:linear}, is valid for
any monotonic penalty.\!\footnote{For the sake of completeness, we revisit this algorithm in supplementary material.}
The algorithm runs in $\bigO{k^2 \log^2 n}$ time.
Thus, the total running time to obtain a segmentation with $(1 + \epsilon)$ approximation guarantee
is
\[
    \bigO{k^2 \log^2(n) + k^3 \log(k) \log(n) + k^3 \epsilon^{-2} \log(n)}\quad.
\]

\section{Strongly polynomial and linear-time scheme for cumulative segmentation}\label{sec:allseg}

In this section we present a strongly polynomial algorithm that approximates
\prballseg in $\bigO{nk^2 / \epsilon}$ time.  Let $o[i, \ell]$ be the cost for
an optimal $\ell$-segmentation covering $[1, i]$.  The optimal segmentation is
computed with a dynamic program based on the identity
\[
	o[i, \ell] = \min_{j \leq i} o[j, \ell - 1] + \score{j, i}\quad.
\]
The integer $j$ yielding the optimal
cost will be the starting point of the last segment in an optimal
segmentation.  To speed-up the discovery of $j$, we will not test every $j \leq i$,
but instead we will use a small set of candidates, say $A$.  So, instead of
computing a single entry in $\bigO{n}$ time, we only need $\bigO{\abs{A}}$
time.

The set $A$ depends on $i$ and $\ell$, and we update it as we change $i$ and
$\ell$.  The key point here is to keep $A$ very small, in fact, $\abs{A} \in
\bigO{k / \epsilon}$, while having enough entries to yield the approximation
guarantee.

Our approach works as follows (see Algorithm~\ref{alg:alldp} for the pseudo-code): 
the algorithm loops over $\ell$ and $i$ with $i$ being the inner for-loop, and maintains a 
set of candidates $A$. There are 3 main steps inside the inner for-loop:

(1)
Test the current candidates $A$.

(2)
Test $a = 1 + \max A$, and add $a$ to $A$. Repeat this step until $a = i$
or if
the current score $s[i, \ell]$ becomes smaller than $s[a, \ell - 1]$.
In the latter case, we can stop because $s[a', \ell - 1] \geq s[i, \ell]$
for any $a' > a$, so we know that there are no candidates that can improve $s[i, \ell]$.

(3)
For every consecutive triplet $a_j, a_{j + 1}, a_{j + 2} \in A$
such that $s[a_{j + 2}, \ell - 1] - s[a_j, \ell - 1] \leq s[i, \ell] \frac{\epsilon}{k + \ell \epsilon}$,
remove $a_j$ (see Algorithm~\ref{alg:sparsify}).
This will keep $\abs{A}$ small for the next round while yielding the approximation guarantee.

\begin{algorithm}[ht!]
	$s[i, 1] \define \score{1, i}$ for $i = 1, \ldots, n$\;
	\caption{$\algalldp(k, \epsilon)$, computes a table $s$ such that $s[i, \ell]$ is a $(1 + \epsilon)$ approximation of the cost of an optimal $\ell$-segmentation covering $[1, i]$.}
	\label{alg:alldp}
	\ForEach {$\ell = 2, \ldots, k$} {
		$A \define \set{1}$\;
		\ForEach {$i = 1, \ldots, n$} {
			$s[i, \ell] \define \min_{a \in A} s[a, \ell - 1] + \score{a, i}$\nllabel{alg:dpmin}\;
			$a \define 1 + \max A$\;
			\While {$a \leq i$ \AND $s[a, \ell - 1] \leq s[i, \ell]$} {
				$s[i, \ell] \define \min (s[i, \ell],\, s[a, \ell - 1] + \score{a, i})$\nllabel{alg:dpmin2}\;
				insert $a$ to $A$\;
				$a \define a + 1$\;
			}
			$A \define \algsparsify(A, s[i, \ell]\frac{\epsilon}{k + \ell\epsilon}, \ell)$\nllabel{alg:callsparse}\;
		}
	}
	\Return $s$\;
\end{algorithm}

\begin{algorithm}[ht!]
	\caption{$\algsparsify(A, \delta, \ell)$, sparsifies $A = a_1, \ldots, a_{\abs{A}}$ using $s[a_i, \ell - 1]$ and $\delta$.}
	\label{alg:sparsify}
	$j \define 1$\;
	\While {$j \leq \abs{A} - 2$} {
		\uIf {$s[a_{j + 2}, \ell - 1] -  s[a_j, \ell - 1] \leq \delta$} {
			remove $a_{j + 1}$ from $A$, and update the indices\;
		}
		\Else {
			$j \define j + 1$\;
		}
	}
\end{algorithm}

Our next step is to prove the correctness of \algalldp.

\begin{proposition}
\label{prop:alldpcorrect}
Let $o[i, \ell]$ be the cost of an optimal $\ell$-segmentation covering $[1, i]$.
Let $s = \algalldp(k, \epsilon)$ be the solution returned by the approximation algorithm. 
Then for any $i = 1, \ldots, n$ and $\ell = 1, \ldots, k$,
\[
	s[i, \ell] \leq o[i, \ell]\pr{1 + \frac{\epsilon \ell }{k}}\quad.
\]
\end{proposition}

Before proving the claim, let us introduce some notation that will be used throughout the remainder of the section.
Assume that $\ell$ is fixed, and
let $A_i$ be the set $A$ at the beginning 
of the $i$th round of \algalldp. Let $m_i = \max A_i$. 
The entries in $A_i$ are always sorted, and we will often refer to these entries as $a_j$.
Let us write $A_i'$ to be the set $A$ which is given to \algsparsify during the $i$th round. 

To prove the claim, we need two lemmas. 
First we show that the score is increasing as a function of $i$.

\begin{lemma}
\label{lem:monotone}
The score is monotone,
$s[i, \ell] \leq s[i + 1, \ell]$.
\end{lemma}

\begin{proof}
We will prove the lemma by induction on $\ell$. The $\ell = 1$ case holds since $s[i, \ell] = \score{1, i}$.
Assume that the lemma holds for $\ell - 1$.

Algorithm~\ref{alg:alldp}
uses the indices in $A'_i$ 
(Lines~\ref{alg:dpmin} and \ref{alg:dpmin2})
for computing $s[i, \ell]$, 
that is,
\begin{equation}
\label{eq:scorebound}
	s[i, \ell] = \min_{a \in A'_i} s[a, \ell - 1] + \score{a, i}\quad.
\end{equation}
Let $B$ be the $\ell$-segmentation corresponding to the cost $s[i + 1, \ell]$.
Let $b$ be the starting point of the last segment of $B$.  
If $b \leq m_{i + 1}$, then $b$ was selected during Line~\ref{alg:dpmin},
that is, $b \in A_{i + 1} \subseteq A'_i$, so Eq.~\ref{eq:scorebound} implies that
\[
\begin{split}
	s[i, \ell] & \leq s[b, \ell - 1] + \score{b, i} \\
	& \leq s[b, \ell - 1] + \score{b, i + 1} = s[i + 1, \ell]\quad.
\end{split}
\]
On the other hand, if $b > m_{i + 1}$, then, due to the while-loop in \algalldp, either
\begin{equation}
\label{eq:boundsi1}
	s[i, \ell] < s[m_{i + 1} + 1, \ell - 1]
\end{equation}
or $i = m_{i + 1} \in A'_i$, and so Eq.~\ref{eq:scorebound} implies that
\begin{equation}
\label{eq:boundsi2}
	s[i, \ell] \leq s[i, \ell - 1] = s[m_{i + 1}, \ell - 1]\quad.
\end{equation}
Due to the induction hypothesis on $\ell$, the right-hand sides of Eqs.~\ref{eq:boundsi1}--\ref{eq:boundsi2}
are bound by $s[b, \ell - 1]$. Since,
\[
	s[b, \ell - 1]  \leq s[b, \ell - 1] +  \score{b, i + 1}  = s[i + 1, \ell],
\]
we have proved the claim.
\end{proof}

The second lemma essentially states that $A_i$ is dense enough to yield an approximation.
To state the lemma, let $\delta_i = s[i, \ell] \times \epsilon / (k + \ell \epsilon)$ be the value of $\delta$
in \algsparsify during the $i$th round. For simplicity, we also define $\delta_0 = 0$.

\begin{lemma}
\label{lem:sparsify}
For every $b \in [1, m_i]$, there is $a_j \in A_i$ s.t.
\[
	s[a_j, \ell - 1] + \score{a_j, i}  \leq s[b, \ell - 1] + \score{b, i} + \delta_{i - 1}\quad.
\]
\end{lemma}

\begin{proof}
We say that a sorted list of indices $X = x_1, \ldots, x_{\abs{X}}$
with $x_1 = 1$
is
$\delta$-dense if $s[x_j, \ell - 1] \leq \delta + s[x_{j - 1}, \ell - 1]$ or
$x_j = x_{j - 1} + 1$ for any $x_j \in X_i$. Note that if $X$ is $\delta$-dense,
then so is $\algsparsify(X, \delta, \ell)$ due to the if-condition in Algorithm~\ref{alg:sparsify}.

We claim that $A_i$ is $\delta_{i - 1}$-dense.
We will prove this by induction on $i$.  This is vacuosly true for $i = 1$.
Assume that $A_{i - 1}$ is $\delta_{i - 2}$-dense.
It is trivial to see that $A'_{i - 1} = A_{i - 1} \cup [m_{i - 1} + 1, m_i]$,
thus $A'_{i - 1}$ is $\delta_{i - 2}$-dense.
Since $\delta_i = s[i, \ell] \times \epsilon / (k + \ell \epsilon)$,
Lemma~\ref{lem:monotone} guarantees that $\delta_{i - 2} \leq \delta_{i - 1}$,
so $A'_{i - 1}$ is $\delta_{i - 1}$-dense. Finally, $A_i$ is $\delta_{i - 1}$-dense, since $A_{i} = \algsparsify(A'_{i - 1}, \delta_{i - 1}, \ell)$.

Let $a_j \in A_i$ be the smallest index that is larger than or equal to $b$.
If $b = a_j$, then the lemma follows immediately. If $b < a_j$,
then  $j > 1$ and
\[
    s[a_j, \ell - 1] \leq \delta_{i - 1} + s[a_{j - 1}, \ell - 1] \leq \delta_{i - 1} + s[b, \ell - 1],
\]
where 
the second inequality follows from the fact that $b > a_{j - 1}$ and Lemma~\ref{lem:monotone}.

The result follows as $\score{a, i} \leq  \score{b, i}$, for any $a \geq b$.
\end{proof}

We can now prove the main result.

\begin{proof}[Proof of Proposition~\ref{prop:alldpcorrect}]
Write $\gamma = 1 + \frac{\epsilon(\ell - 1)}{k}$.

We will prove the claim using induction on $\ell$ and $i$. The $\ell = 1$ or $i = 1$ cases is trivial.

To prove the general case,
let $B$ be an optimal $\ell$-segmentation covering $[1, i]$, and
let $b$ be the starting point of the last segment of $B$.  
We consider two cases.

\emph{Case (i):}
Assume that $b \leq m_i$, then according to Lemma~\ref{lem:sparsify}, there is $a \in A_i$ for which
\begin{equation}
\label{eq:boundsparsify}
	s[a, \ell - 1] + \score{a, i}  \leq s[b, \ell - 1] + \score{b, i} + \delta_{i - 1}\quad.
\end{equation}
Recall that $\delta_{i - 1} = s[i - 1, \ell] \frac{\epsilon}{k + \ell\epsilon}$, due to Line~\ref{alg:callsparse} in Alg.~\ref{alg:alldp}.
Using the induction hypothesis on $i$,
we can bound $\delta_{i - 1}$, 
\begin{equation}
\label{eq:bounddelta}
\begin{split}
	\delta_{i - 1} & = s[i - 1, \ell] \frac{\epsilon}{k + \ell\epsilon}
	\leq o[i - 1, \ell]\pr{1 + \epsilon\frac{\ell}{k}} \frac{\epsilon}{k + \ell\epsilon} \\
	& = o[i - 1, \ell]\frac{\epsilon}{k}
	\leq o[i, \ell]\frac{\epsilon}{k}\quad.
\end{split}
\end{equation}
We can now combine the previous inequalities and the induction hypothesis on $\ell$, which gives us
\begin{align*}
	s[i, \ell] & \leq \min_{x \in A_i} s[x, \ell - 1] + \score{x, i}  &\text{\llap{(Line~\ref{alg:dpmin} in Alg.~\ref{alg:alldp})}}\\
	& \leq s[a, \ell - 1] + \score{a, i}  & (a \in A_i)\\
	& \leq s[b, \ell - 1] + \score{b, i} + o[i, \ell]\frac{\epsilon}{k}  & \text{\llap{(Eqs.~\ref{eq:boundsparsify}--\ref{eq:bounddelta})}}\\
	& \leq o[b, \ell - 1]\gamma + \score{b, i} + o[i, \ell]\frac{\epsilon}{k}  & \text{\llap{(induction)}}  \\
	& \leq o[b, \ell - 1]\gamma + \score{b, i}\gamma + o[i, \ell]\frac{\epsilon}{k}  & (\gamma \geq 1) \\
	& = o[i, \ell]\gamma + o[i, \ell]\frac{\epsilon}{k}  
	 = o[i, \ell]\pr{1 + \frac{\epsilon\ell}{k}}\quad.  && \\
\end{align*}

\emph{Case (ii):}
Assume that $b > m_i$. If $b \in A_i'$, then
\begin{equation}
\label{eq:boundb1}
	s[i, \ell] = \min_{x \in A_i'} s[x, \ell - 1] + \score{x, i} \leq s[b, \ell - 1] + \score{b, i}\quad.
\end{equation}
Assume $b \notin A_i'$. This is only possible if the second condition in the while-loop failed, that is,
there is $a \leq b$ with
\begin{equation}
\label{eq:boundb2}
	s[i, \ell] \leq s[a, \ell - 1] \leq s[b, \ell - 1]  \leq s[b, \ell - 1] + \score{b, i}\quad.
\end{equation}
In both cases, the induction hypothesis on $\ell$ gives us
\begin{align*}
	s[i, \ell] & \leq s[b, \ell - 1] + \score{b, i}  & (\text{Eqs.~\ref{eq:boundb1}--\ref{eq:boundb2}})\\
	& \leq o[b, \ell - 1]\gamma + \score{b, i} & (\text{induction})\\
	& \leq o[b, \ell - 1]\gamma + \score{b, i} \gamma & (\gamma \geq 1)\\
	& = o[i, \ell]\gamma 
	 \leq o[i, \ell]\pr{1 + \frac{\epsilon\ell}{k}}\quad. 
\end{align*}
This proves the induction step and the proposition.
\end{proof}

Finally, let us prove the running time of \algalldp.

\begin{proposition}
\label{prop:alldptime}
$\algalldp(k, \epsilon)$ needs
$\bigO{\frac{k^2}{\epsilon} n}$ time.
\end{proposition}

We will adopt the same notation as with the proof of Proposition~\ref{prop:alldpcorrect}.
For simplicity, we also define $s[0, \ell] = 0$ for any $\ell = 1, \ldots, k$.

To prove the result we need two lemmas. 
The first lemma will be used to prove the second lemma.

\begin{lemma}
\label{lem:candbound}
$s[m_i, \ell - 1] \leq s[i - 1, \ell]$, where
$m_i = \max A_i$.
\end{lemma}

\begin{proof}
We prove the lemma using induction on $i$. 
The case $i = 1$ is trivial since $A_1 = \set{1}$ and $s[1, \ell - 1] = (\ell - 1)p(1, 1) = 0 = s[0, \ell]$.

Assume that the claim holds for $i - 1$.
Now, Lemma~\ref{lem:monotone} guarantees that
$s[m_{i - 1}, \ell - 1] \leq s[i - 2, \ell] \leq s[i - 1, \ell]$.
The while-loop condition in \algalldp guarantees that $s[m, \ell - 1] \leq s[i - 1, \ell]$, for $m = \max A'_{i - 1}$.
Since \algsparsify never deletes the last element, we have $m = m_i$, which proves the lemma.
\end{proof}

Next, we bound the number of items in $A_i$.

\begin{lemma}
\label{lem:candsize}
$\abs{A_i} \leq 2 + 2(k + \ell \epsilon) / \epsilon  \in \bigO{k / \epsilon}$.
\end{lemma}

\begin{proof}
Let $r = \lfloor (\abs{A_i}  + 1)  / 2\rfloor$ be the number of entries in $A_i$ with odd indices.
Due to \algsparsify, $A_i$ cannot contain two items, say $a_j$ and $a_{j + 2}$, such that
$s[a_{j + 2}, \ell - 1] \leq s[a_j, \ell - 1] + \delta_{i - 1}$.
Using this inequality on entries with odd incides we conclude that
	$s[m_i, \ell - 1] \geq (r - 1) \delta_{i - 1}$.
Lemma~\ref{lem:candbound} allows us to bound $r$ with
\[
	r - 1 \leq \frac{s[m_i, \ell - 1]}{\delta_{i - 1}} \leq  \frac{s[i - 1, \ell]}{\delta_{i - 1}} = \frac{k + \ell \epsilon}{\epsilon} \in \bigO{\frac{k}{\epsilon}}\quad.
\]
Since $\abs{A_i} \leq 2r$, the lemma follows.
\end{proof}

We can now prove the main claim.

\begin{proof}[Proof of Proposition~\ref{prop:alldptime}]
Computing $s[i, \ell]$ in \algalldp requires at most $\abs{A'_i}$ comparisons,
and \algsparsify requires at most $\abs{A'_i}$ while-loop iterations.
We have $\abs{A'_i} = \abs{A_i} + m_{i + 1} - m_i$.
For a fixed $\ell$, Lemma~\ref{lem:candsize} implies that the total number of comparisons is 
\[
\begin{split}
	\sum_{i = 1}^n \abs{A'_i} & = \sum_{i = 1}^n \abs{A_i} + m_{i + 1} - m_i \\
	& \leq n + \sum_{i = 1}^n \abs{A_i} \in \bigO{\frac{n k}{\epsilon}}\quad.
\end{split}
\]
The result follows since $\ell = 1, \ldots, k$.
\end{proof}

The idea behind this
approach is similar to the algorithm suggested by~\citet{guha:01:estimate}.
The main difference is how the candidate list is formed: 
\citet{guha:01:estimate} constructed the candidate list
by only adding new entries to it, whereas we are also deleting the values
allowing us to keep $A$ much smaller.

\section{A linear-time algorithm for\\cumulative maximum segmentation}\label{sec:allmaxseg}

In previous section, we presented a technique for computing a cumulative
segmentation. Our final contribution is a fast algorithm to solve cumulative
version of the maximum segmentation problem \prbmaxseg.

\begin{problem}[\prballmaxseg]
\label{prb:allmax}
Given a penalty function $\score{}$ and an integer $k$, 
find an $\ell$-segmentation $B$ covering $[1, i]$ that minimizes 
\[
	\scoremax{B} = \max_{1 \leq j \leq \ell} \score{b_{j - 1}, b_j},
\]
for every $i = 1, \ldots, n$ and every $\ell = 1, \ldots, k$.
\end{problem}

\begin{algorithm}[ht!]
\caption{$\algfmslow(k)$, solves \prballmaxseg.}
\label{alg:fmslow}
$b_i \define 1$, for $i = 0, \ldots, k$\;
$b_{k + 1} \define n$\tcpas*{sentinel, discarded in the end}
$\tau \define 0$\;
$s[1, \ell] \define 0$, for $\ell = 1, \ldots, k$\;
\While {$b_1 \leq n$} {
	$\ell \define \arg \min\limits_{1 \leq j \leq k}  \set{\score{b_{j - 1}, b_{j} + 1} \mid b_j < b_{j + 1}}$\; 
	$b_\ell \define b_\ell + 1$\;
	$\tau \define \max(\tau, \score{b_{\ell - 1}, b_{\ell}})$\;
	$s[b_\ell, \ell] \define \tau$\;
}
\Return $s$\;
\end{algorithm}

We should point out that we can apply the algorithm by~\citet{guha:07:linear} used to solve \prbmaxseg for
every $i$ and $\ell$, which would give us the computational complexity of $\bigO{nk^3\log^2 n}$.
Alternatively, we can use an algorithm given by~\citet{DBLP:conf/vldb/GuhaSW04} that provides us 
with the computational complexity of $\bigO{kn \log n}$.\!\footnote{Analysis by
\citet{DBLP:conf/vldb/GuhaSW04} states that the computational complexity is $\bigO{kn\log^2n}$ but the additional $\log n$
term is due to the $\bigO{\log n}$ oracle used to compute the segment cost.}
We will present a faster algorithm, running in $\bigO{n k \log k}$ time,
making the algorithm linear-time with respect to $n$.

Our algorihm works as follows (see Algorithm~\ref{alg:fmslow}): 
We start with a $k$-segmentation $B = (b_0 = 1, \ldots, b_k = 1)$.
It turns out that we can choose $b_\ell$ 
so that an $\ell$-segmentation, say $B' = (b_0' = b_0, \ldots, b_{\ell - 1}' = b_{\ell - 1},
b_\ell' = b_\ell +1)$ is an optimal $\ell$-segmentation covering $[1, b_\ell']$.
We increase  value of $b_\ell$ by 1, and repeat
until there are no indices that cannot be moved right, that is, $b_1 = n$.
At this point, every $b_\ell$ has visited every integer
between $1$ and $n$, so we have discovered all optimal segmentations.

To guarantee that $(b_0', \ldots, b_\ell')$ is optimal, we need to select
$\ell$ carefully. Here, we choose $\ell$ to be the index minimizing
$\score{b_{\ell - 1}, b_{\ell} + 1}$.  The main reason for choosing such a
value comes from the next lemma.

\begin{lemma}
\label{lem:maxbound}
Assume an $\ell$-segmentation $B = b_0, \ldots, b_\ell$ covering $[1, i]$. Let $\tau$ be such that
$\tau \leq \score{b_{c - 1}, b_c + 1}$, for every $c = 1, \ldots, \ell$.
Let $B'$ be an $\ell$-segmentation covering $[1, j]$ with $j > i$. Then $\scoremax{B'} \geq \tau$.
\end{lemma}

\begin{proof}
Assume that $\scoremax{B'} < \tau$.
We claim that $b'_c \leq b_c$ for every $c = 1, \ldots, \ell$, and
this claim leads to $j = b_\ell' \leq b_\ell = i$, which is a contradiction.

We will prove the claim by induction.
To prove the induction base, note that
$b'_1 \leq b_1$, as otherwise $\score{b_{0}', b_1'} \geq \score{b_{0}, b_1 + 1} \geq \tau$.
To prove the induction step, assume that $b'_{c - 1} \leq b_{c - 1}$.
Then, $b'_c \leq b_c$, as otherwise $\score{b_{c - 1}', b_c'} \geq \score{b_{c - 1}, b_c + 1} \geq \tau$.
\end{proof}

The lemma can be used as follows:
Assume a current $k$-segmentation $B = (b_0 = 1, \ldots, b_k =
1)$, and let $\ell$ be the index minimizing $\tau = \score{b_{\ell - 1},
b_{\ell} + 1}$.  Let $B' = (b_0' = b_0, \ldots, b_{\ell - 1}' = b_{\ell - 1},
b_\ell' = b_\ell +1)$ be the resulting $\ell$-segmentation.  Assume that
$\scoremax{B'} = \tau$. Lemma~\ref{lem:maxbound} implies that there is no
segmentation with cost smaller than $\tau$ covering $[1, b_\ell']$. This makes
automatically $B'$ optimal. The proof for the case $\scoremax{B'} < \tau$ is
more intricate, and it is handled in the next proposition.

\begin{proposition}
\label{prop:maxcorrect}
$\algfmslow(k)$ solves \prballmaxseg.
\end{proposition}

\begin{proof} 
Define $\tau_i$ to be the value of $\tau$ at the end of the $i$th
iteration.  Let also $B^i = b^i_1, \ldots, b^i_k$ be the values of
$b_1, \ldots, b_k$ at the end of the $i$th iteration.

Fix $i$ and assume that during the $i$th iteration we updated $b_\ell$.
Let $\tau^*$ be the optimal cost for $\ell$-segmentation covering $[1, b_\ell^i]$.
Since $b^i_0, \ldots, b^i_\ell$ covers $[1, b_\ell^i]$ and the cost of individual segments
is bounded by $\tau_i$, we have $\tau^* \leq \tau_i$.

To prove the optimality of  $b^i_0, \ldots, b^i_\ell$, we need to show that $\tau^* \geq \tau_i$.
Let $j$ be the largest index such that $\tau_j < \tau_i$.
If such value does not exist, then $\tau_i = 0 \leq \tau^*$.
Clearly, $j < i$.
Assume that we update $b_{\ell'}$ during the $(j + 1)$th iteration.
Then, by definition of $\ell'$, 
\[
	\score{b^j_{c - 1}, b^j_c + 1} \geq \score{b_{\ell' - 1}^j, b_{\ell'}^j + 1} = \tau_{j + 1} = \tau_i,
\]
for any $c = 1, \ldots, \ell$.
Since $b_{\ell}^j \leq b_{\ell}^{i - 1} < b_{\ell}^i$, Lemma~\ref{lem:maxbound} implies that an $\ell$-segmentation
covering $[1, b_\ell^i]$ must have a cost of a least $\tau_i$. This proves the proposition.
\end{proof}

We finish with the computational complexity analysis.

\begin{proposition}
\algfmslow runs in $\bigO{nk \log k}$ time.
\end{proposition}

\begin{proof}
Since $b_j$ can move only to the right, the while-loop is evaluated at most $\bigO{kn}$ times.
Computing $j$ at each iteration can be done by maintaining a priority queue of size $\bigO{k}$.
After updating $b_j$, updating the queue can be done in $\bigO{\log k}$ time.
\end{proof}

\bibliographystyle{abbrvnat}
\bibliography{bibliography}

\end{document}


\begin{frontmatter}
\title{Supplementary material for strongly polynomial efficient approximation scheme for segmentation}
\author{Nikolaj Tatti}
\address{F-Secure, HIIT, Aalto University, Finland}
\ead{nikolaj.tatti@aalto.fi}

\begin{abstract}
In this supplementary material we revisit an algorithm proposed by~\citet{guha:07:linear},
and show that this algorithm can be used to solve the maximum segmentation problem.
\end{abstract}

\end{frontmatter}



\section{Maximum segmentation}\label{sec:seg}

\begin{problem}[\prbmaxseg]
\label{prb:max}
Given a penalty $\score{}$ and an integer $k$, 
find a $k$-segmentation $B$ covering $[1, n]$ that minimizes 
\[
	\scoremax{B} = \max_{1 \leq j \leq k} \score{b_{j - 1}, b_j}\quad.
\]
\end{problem}

This optimization problem can be solved with an algorithm given by~\citet{guha:07:linear}.
This algorithm is designed to solve an instance of \prbmaxseg, where the penalty function is $L_\infty$ 
is
\[
	\score{a, b} = \min_{\mu} \sum_{i = a}^{b - 1} |x_i - \mu|\quad.
\]
Luckily, the same algorithm and the proof of correctness, see Lemma~3~in~\citep{guha:07:linear}, is valid for
any monotonic penalty.

For the sake of completeness we present this algorithm in Algorithms~\ref{alg:long}--\ref{alg:fmfast}.

The main idea is based on the following observation. Let $B$ be the optimal
segmentation with the cost of $\scoremax{B} = \tau$.  Then $\score{b_0, b_1} =
\tau$ or $\scoremax{B'} = \tau$, where $B' = b_1, \ldots, b_k$ is a $(k -
1)$-segmentation covering $[b_1, n]$.
In the first case, we can safely assume
that $b_1 = c$, where $c$ is the smallest index for which there is a $(k - 1)$-segmentation
covering $[b_1, n]$ with a maximum cost of $\score{b_0, b_1}$.
In the second
case, we can safely assume that $b_1 = c - 1$, that is, the largest index for
which there is \emph{no} $(k - 1)$-segmentation covering $[b_1, n]$ with a
maximum cost of $\score{b_0, b_1}$.

This gives rise to the main loop: compute
$c$, and record $\Delta_1 = \score{b_0, c}$, then recurse and discover the best
$(k - 1)$ segmentation covering $[c - 1, n]$, with a cost of, say, $\Delta_2$.
Then, the correct cost is $\min(\Delta_1, \Delta_2)$. For the complete proof of
correctness see Lemma~3~in~\citep{guha:07:linear}.\!\footnote{In pseudo-code
given in~\cite{guha:07:linear}, an incorrect step, $i \define c$, is used.
However, in the proof of correctness, Lemma~3, correct value is used.}

For completeness, we provide the proof of correctness.
In order to do so, let us first define
\[
	f(b; i, k) = [\alglong(b, k - 1, \score{i, b}) = n],
\]
returning true or false depending whether the statement inside the brackets is valid.

\begin{algorithm}[ht]
\caption{$\alglong(b, k, \tau)$ computes the largest index that can be reached with a $k$-segmentation
starting from $b$ with a cost $\scoremax{} \leq \tau$.}
\label{alg:long}
	\ForEach{$\ell = 1, \ldots, k$} {
		$b \define \max_j \set{b \leq j \leq n \mid \score{b, j} \leq \tau}$\;
		\nonl\hfill\tcpas{use binary search}
	}
\Return $b$\;
\end{algorithm}

\begin{algorithm}[ht!]
\caption{$\algfmfast(k)$, computes the cost of $k$-segmentation solving \prbmaxseg.}
\label{alg:fmfast}
    $\Delta \define \infty$\;
	$i \define 1$\;
    \ForEach{$\ell = k, \ldots, 1$}{
        $c \define \min \left\{b \geq i \mid
                \alglong(b, \ell - 1, \score{i, b})  = n \right\}$\;
		\nonl\hfill\tcpas{use binary search}
		$\Delta \define \min (\Delta,  \score{i, c})$\;
		\lIf {$i = c$} {\Return $\Delta$}
		$i \define c - 1$\;
    }
    \Return $\Delta$\;
\end{algorithm}

\begin{proposition}
$f(b; i, k)$ returns true if and only if there is a $k$-segmentation $B$
covering $[i, n]$ with $\scoremax{B} \leq \score{i, b}$.
\end{proposition}

\begin{proof}
The only if direction is trivial.

To prove the if direction, assume that $B$ satisfies the conditions of the
proposition, and let $C$ be the segmentation constructed by $\alglong(b, k - 1,
\score{i, b})$.  Note that $b_1 \leq c_1$.  If $b_1 < c_1$, then we can
move $b_1$ to the right, without violating the conditions. Thus, we can safely
assume that $b_1 = c_1$. By doing this recursively, we can safely assume that $b_i = c_i$.
This guarantees that $f$ returns true.
\end{proof}

\begin{proposition}
\algfmfast returns the cost of the optimal segmentation.
\end{proposition}

\begin{proof}
Let us write $i_\ell$, $c_\ell$ to be the values of variables
during the foreach-loop of \algfmfast. Note that $\ell$ goes from $k$ to $1$,
or terminated early.

Write $\tau_\ell = \score{i_\ell, c_\ell}$.
Define $C_\ell$ to be the segmentation discovered by $\alglong(c_\ell, \ell - 1,
\score{i_\ell, c_\ell})$ preceding by the segment $(i_\ell, c_\ell)$.

Let $B_\ell$ be the optimal segmentation
covering $[i_\ell, n]$ with the cost of $\rho_\ell$.
Define $C'_\ell$ to be $B_{\ell - 1}$ preceding by the segment $(i_\ell, c_\ell - 1) = (i_{\ell}, i_{\ell - 1})$.

By definition of $c_\ell$, $\tau_\ell = \score{i_\ell, c_\ell} = \scoremax{C_\ell}$.
Since $C_\ell$ covers $[i_\ell, n]$, we also have
\[
	\tau_\ell = \score{i_\ell, c_\ell} = \scoremax{C_\ell} \geq \rho_\ell\quad.
\]

By the minimality of $c_\ell$, 
$\score{i_\ell, c_\ell - 1} < \scoremax{B_{\ell - 1}}$, and 
since $C_\ell'$ covers $[i_\ell, n]$, we also have
\[
	\rho_{\ell - 1} = \scoremax{B_{\ell - 1}} =  \scoremax{C_\ell'} \geq \rho_\ell\quad.
\]

Assume that the for-loop is terminated early due to the $i_\ell = c_\ell$ condition.
Then $0 = \tau_\ell \geq \rho_\ell \geq \rho_k$, that is, there is a $k$-segmentation covering
$[1, n]$ with zero cost. Since the algorithm outputs $0$, this proves the special case,
so we can safely assume that the for-loop is not terminated early.

Define $\eta_\ell = \min_{j \leq \ell} \tau_j$.
Note that $\eta_k$ is the output of the algorithm, so
to prove the result,
we claim that $\eta_\ell = \rho_\ell$ which we will prove by induction on $\ell$.
The case $\ell = 1$ is trivial. 

Fix $\ell$ and let $b$ be the ending point of the first segment in $B_\ell$.
We consider two cases.

Assume that $f(b; i_\ell, \ell)$ is true.
Then we must have $c_\ell \leq b$.
Thus, $\tau_\ell = \score{i_\ell, c_\ell} \leq \score{i_\ell, b} \leq \rho_\ell$.
So, $\tau_\ell = \rho_\ell$. The induction hypothesis states that $\eta_{\ell - 1} = \rho_{\ell - 1} \geq \rho_\ell$.
Hence, $\eta_\ell = \min(\eta_{\ell - 1}, \tau_\ell) = \rho_\ell$.

Assume that $f(b; i_\ell, \ell)$ is false. Then we must have $c - 1 \geq b$.
In other words, $B_{\ell - 1}$ need to cover less than the remaining segments of $B_\ell$.
Thus $\rho_{\ell - 1} \leq \rho_\ell$. Since, $\tau_\ell \geq \rho_\ell$ and 
$\rho_{\ell - 1} \geq \rho_\ell$, the induction hypothesis implies that
\[
	\eta_\ell = \min (\tau_\ell, \eta_{\ell - 1}) = \min (\tau_\ell, \rho_{\ell - 1}) = \min (\tau_\ell, \rho_{\ell}) = \rho_\ell,
\]
this proves the induction step.
\end{proof}

To show the computational complexity, 
note that for a fixed $\ell$ we need
$\bigO{\log n}$ evaluations of \alglong to compute $c$, each
evaluation requiring $\bigO{k \log n}$ time. Thus, computing a single $c$ requires $\bigO{k \log^2 n}$.
We need to repeat this $\bigO{k}$ times, which gives us a total running time of $\bigO{k^2 \log^2 n}$.

\bibliographystyle{abbrvnat}
\bibliography{bibliography}